\documentclass{article}
\usepackage{wright}
\usepackage{subcaption}
\usepackage{qcircuit}
\pdfoutput=1

\AtEveryBibitem{ 
    \clearfield{day}
    \clearfield{month}
    \clearfield{series}
    \clearfield{venue}
    \clearname{editor}
    \clearlist{publisher}
    \clearlist{location} 
    \clearfield{venue}
    \clearfield{issn}
    \clearfield{isbn}
    \clearfield{urldate}
    \clearfield{eventdate}
    \clearfield{pages}
    \clearfield{number}
    \clearfield{volume}
}

\title{Are controlled unitaries helpful?}
\author{
    Ewin Tang\thanks{UC Berkeley. \texttt{\{ewin,jswright\}@berkeley.edu}} 
    \and John Wright\footnotemark[1]
}

\date{}

\begin{document}

\maketitle


\begin{abstract}
Many quantum algorithms, to compute some property of a unitary $U$, require access not just to $U$, but to $\controlled U$, the unitary with a control qubit.
We show that having access to $\controlled U$ does not help for a large class of quantum problems.
For a quantum circuit which uses $\controlled U$ and $\controlled U^\dagger$ and outputs $\ket{\psi(U)}$, we show how to ``decontrol'' the circuit into one which uses only $U$ and $U^\dagger$ and outputs $\ket{\psi(\varphi U)}$ for a uniformly random phase $\varphi$, with a small amount of time and space overhead.
When we only care about the output state up to a global phase on $U$, then the decontrolled circuit suffices.
Stated differently, $\controlled U$ is only helpful because it contains global phase information about $U$.

A version of our procedure is described in an appendix of Sheridan, Maslov, and Mosca~\cite{smm09}.
Our goal with this work is to popularize this result by generalizing it and investigating its implications, in order to counter negative results in the literature which might lead one to believe that decontrolling is not possible.
As an application, we give a simple proof for the existence of unitary ensembles which are pseudorandom under access to $U$, $U^\dagger$, $\controlled U$, and $\controlled U^\dagger$.
\end{abstract}
%
\hypersetup{linktocpage}
\tableofcontents
%

\section{Introduction}


\emph{Unitary problems} have revealed themselves to be fundamental to how we think about quantum computation.
Over the past quarter-century, the most important quantum algorithms have all been abstracted away into more general subroutines which capture their key algorithmic ideas.
Today, we understand Grover's algorithm for unstructured search~\cite{grover96} as a special case of amplitude amplification~\cite{BHMT02}; Shor's algorithm for factoring~\cite{shor1994Factoring} is phase estimation~\cite{kitaev95} wrapped in a number-theoretic cloak; and Hamiltonian simulation~\cite{lloyd96} and Harrow--Hassidim--Lloyd's algorithm for solving sparse linear systems~\cite{hhl09} both fall under the large umbrella of quantum singular value transformation~\cite{gslw18}.
These general subroutines are best described as solving unitary problems: the input is an arbitrary quantum process $U$ which we have black-box access to, and the output is some function of $U$, whether it be a classical bitstring, a quantum state, or a quantum circuit itself.
As a consequence, we have seen a shift towards treating the unitary problem as part of the foundation for quantum computation, particularly on the side of complexity, where recent works have begun developing a complexity theory of unitary problems~\cite{bempqy23,LMW24,zhandry25}.


When formalizing the notion of a unitary problem, the first question we run into is what it means to have access to an input unitary.
A curious feature of many quantum algorithms is that they require more than just access to $U$.\footnote{
    Amplitude amplification and QSVT use queries to $U$ and $U^\dagger$.
    Phase estimation uses queries to $\controlled U$.
}
Most commonly, in addition to the unitary $U$ itself, quantum subroutines often require access to its inverse, $U^\dagger$, and/or its controlled version, $\controlled U$:
\begin{align*}
    \controlled U = \proj{0} \otimes I + \proj{1} \otimes U
    = \begin{bmatrix}
        I & \\ & U
    \end{bmatrix}.
\end{align*}
But what are the relationships between these different kinds of accesses?
When do we need stronger forms of access?
And when do weaker forms of access suffice?

These questions are of fundamental importance for understanding the broader theory of unitary problems.
But currently, we know surprisingly little about why these different types of access are needed.
As a result, access model subtleties often cause annoyance for researchers in the theory of quantum computation.
Let us consider the main focus of our work,\footnote{
    We discuss the situation regarding $U^\dagger$ in our concurrent work~\cite{tw25a}.
} $\controlled U$.
Issues regarding $\controlled U$ crop up in a few places.


First, we sometimes want lower bounds which hold when given access to $\controlled U$.
When we consider unitary problems, we typically imagine that the input $U$ is instantiated via a quantum circuit.
In this case, by replacing every gate of $U$'s circuit with a small, constant-sized implementation of its controlled version, we can implement $\controlled U$ with a constant factor overhead in gate count.
So, with a scalable quantum computer and a circuit implementation of $U$, it's reasonable to assume that the cost of applying $U$ is similar to that of applying $\controlled U$.
Consequently, if we have a hardness result for a unitary problem (say, hardness of breaking a cryptographic protocol), this hardness result had better hold when given $\controlled U$ as well, as otherwise the lower bound won't apply in real-world situations where we can look inside the black box and see its implementation.
Thus, we need to ask: when does $\controlled U$ not help?

Second, we sometimes want to run algorithms without using $\controlled U$.
When $U$ is instantiated by a long quantum circuit, 
the constant-factor overhead needed to implement $\controlled U$ might be expensive in practice.
Further, in more experimental settings, like those of quantum sensing, metrology, and learning, $U$ comes from nature, so getting access to $\controlled U$ can require a significantly more advanced experimental setup.
Quantum computation is, for now and for the foreseeable future, extremely resource-limited.
So, beyond being a technical curiosity for the quantum circuit designer, requiring special access for quantum circuits imposes genuine limitations on their downstream applications.
If we have an algorithm which uses $\controlled U$, we must ask: can we get away without it?


To understand why we need these stronger forms of access, we ask whether there is a particular kind of unitary problem which is only possible to solve with this strengthened access.
For $\controlled U$, it is well-known that such a problem exists, since controlled access helps to identify the global phase of $U$.
As an example, consider the following problem: given a unitary oracle $U$, decide whether it is the identity $I$ or its negation, $-I$.
With only access to $U$, this problem is impossible, as a quantum circuit querying only $U$ cannot distinguish the difference in global phase between the two cases.
On the other hand, this problem can be solved with just a single query to $\controlled U$.

Many kinds of unitary problems do not care about the global phase of $U$, however. 
In this work, we ask whether global phase is the \emph{only} reason why controlled unitaries are necessary:
\begin{center}
    Is $\controlled U$ ever useful for a problem which is invariant under global phase?
\end{center}
One might guess at first blush that the answer to this question is yes, $\controlled U$ is sometimes useful for phase-invariant problems.
A literature search reveals two forms of evidence for the power of controlled unitaries. 
First, many algorithms seem to need $\controlled U$ despite the global phase of $U$ seeming irrelevant to the problem being solved.
Such situations commonly appear when a subroutine, which requires $\controlled U$ because it depends on global phase, is used for a problem which ultimately doesn't depend on the global phase of $U$~\cite{BHMT02,hlm17}.
Second, if controlled unitaries indeed only help to distinguish global phase, then there should be some generic method to simulate controlled unitary queries with uncontrolled queries in some global phase-invariant sense.
Instead, what we see are broadly negative results, with impossibility results for simulating $\controlled U$ from queries to $U$~\cite{afcb14,gst24}.
Curiously, though, this evidence nevertheless does not rule out the possibility that the answer to our question is no.

\subsection{Results}


We prove that $\controlled U$ does not help for problems which are invariant under global phase, answering our question in the negative.
To do this, we show how to simulate a circuit using $\controlled U$'s with a circuit using only $U$'s, where the simulation circuit outputs the original circuit's output on $\controlled (\varphi U)$ with a random phase $\varphi$.

\paragraph{Notation.}
We denote $\ii = \sqrt{-1}$.
For $q$ a positive integer, we let $\cyc_q \coloneqq \braces{e^{2\pi \ii k / q} \mid k = 0,1,\dots,q-1}$ be the group generated by a $q$-th root of unity.
Abusing notation, we let $\cyc_\infty \coloneqq \braces{e^{\ii \theta} \mid 0 \leq \theta < 2\pi}$ be the complex unit circle.
We use sans serif script to refer to registers of a quantum state: for example, $\ket{\psi}_{\reg{P}}\ket{\phi}_{\reg{RS}}$ denotes the system where $\ket{\psi}$ occupies register $\reg{P}$ and $\ket{\phi}$ occupies registers $\reg{R}$ and $\reg{S}$.
We drop the register notation when the set of registers is clear from context.
For a matrix $M$, $M^*$, $M^\trans$, and $M^\dagger$ denote $M$'s conjugate, transpose, and conjugate transpose, respectively.

\begin{theorem}[Wresting quantum control from a quantum circuit] \label{thm:main}
    Consider a quantum circuit (possibly with ancilla, possibly with measurements and partial traces) which queries $\controlled U$, $\controlled U^\dagger$, $\controlled U^*$, and $\controlled U^{\trans}$ a cumulative total of $n$ times and produces the output state $\rho(U)$.
    
    Then there is a method to convert this circuit into one which outputs $\wh{\rho} = \E_{\phi \sim \cyc_\infty}\bracks{\rho(\phi U)}$ and which replaces every controlled query with its un-controlled variant.
    Further, for every $q > n$, $\wh{\rho} = \E_{\phi \sim \cyc_q}\bracks{\rho(\phi U)}$.

    When $U$ is a unitary on $\log_2(d)$ qubits, this simulation uses $\ceil{\log_2(n)} + 2\log_2(d)$ additional qubits and $\bigO{n(\log(n) + \log(d))}$ additional gates.
    (This gate and space overhead can be reduced; see \cref{prop:less-overhead}.)
\end{theorem}

We found this result quite surprising, as it cuts against an existing body of work on simulating $\controlled U$'s with $U$'s.
It is not possible to simulate a query to $\controlled (\phi U)$ for any $\phi \in \cyc_{\infty}$ given a query to $U$~\cite{afcb14}, because of issues with the global phase on $U$.
This is true even when the phase $\phi$ is allowed to depend (deterministically) on $U$.
This might be concerning: why can't we build the quantum version of an ``if clause'' (if the control qubit is $1$, apply $U$; else, apply $I$) given black-box access to $U$?
Classical if statements can certainly be implemented in this way, so is this a quirk of quantum mechanics?

As it turns out, this barrier is not fundamental: if we don't care about the phase on $U$ in the controlled unitary, with one query to $U$ we can implement the non-unitary channel which corresponds to applying $\controlled (\varphi U)$ for a random $\varphi$.
The key observation is that we can design this implementation so that the random phase $\varphi$ in the controlled unitary is \emph{consistent across multiple queries}, so a circuit using this implementation can expect the same controlled unitary every time it queries it.
A weaker version of \cref{thm:main} is shown by Sheridan, Maslov, and Mosca~\cite{smm09}; the main difference is that, with their version, the random $\varphi$ is chosen from a distribution which is $U$-dependent, whereas our distribution is a uniform phase from some $\cyc_q$.
We discuss this prior work more in \cref{subsec:prior}.
Our goal with this work is to popularize their idea and present it in as wide a scope as possible, in order to illuminate its broader implications.


The specific way to ``decontrol'' the circuit can be seen in \cref{fig:simulation,fig:conversion}.
The implication of this is that if we have a quantum circuit which computes something about $U$ which is invariant under global phase, then we would be perfectly happy with the output state $\rho(\phi U)$ averaged over a random $\phi$.
So, we can use the decontrolled circuit instead of the controlled one to perform our computation.

\subsection{Which problems do not need controlled unitaries?} \label{subsec:control}

Let us make this precise with an example about state preparation unitaries.
We call $U$ a state preparation unitary of a state with density matrix $\sigma$ if $U\ket{0}$ is a purification of $\sigma$, meaning that we can trace out a subsystem of $U\ket{0}$ to get $\sigma$.

\begin{corollary}[Quantum control of state preparation unitaries does not help] \label{cor:state-prep}
    Consider the following distinguishing task on quantum states.
    Let $\calH_0$ and $\calH_1$ be disjoint classes of quantum states.
    Given a state preparation unitary $U$ for a state $\sigma$ in one of these classes, decide which class it came from with success probability $\geq \tfrac{2}{3}$.
    Suppose this task can be solved in $n$ queries, given controlled access to $U$ and its inverse, $\controlled U$ and $\controlled U^\dagger$.
    Then it can also be solved with $n$ queries to $U$ and $U^\dagger$.
\end{corollary}
\begin{proof}
Let $\rho(U)$ be the output of the circuit which solves the distinguishing task using $n$ queries to $\controlled U$ and $\controlled U^\dagger$.
Because the circuit is allowed measurements and partial traces, we can assume that $\rho(U)$ is the output guess: $\rho(U)$ is a random classical bit which identifies the correct class with probability $\geq 2/3$.
Then \cref{thm:main} gives a circuit which uses $n$ queries to $U$ and $U^\dagger$ and outputs $\rho(\varphi U)$ for a random $\varphi \in \cyc_\infty$.
If $U$ is a state preparation unitary for a state $\sigma$, then so is $\varphi U$ for any phase $\varphi \in \cyc_\infty$.
So, this output $\rho(\varphi U)$ also identifies the correct class with probability $\geq 2/3$.
\end{proof}

The consequences of our main theorem can be fairly non-obvious.
In section 5.1.5 of the property testing survey of Montanaro and de Wolf~\cite{MdW16}, they give an algorithm for property testing the commutativity of two unitary matrices $U$ and $V$, which uses $\poly(1/\eps)$ many queries to $\controlled U$ and $\controlled V$.
They ask whether controlled queries are needed~\cite[Question 10]{MdW16}; a corollary of our work is that they are not.
\begin{corollary}[Quantum control does not help for property testing commutativity] \label{cor:commutativity}
    Consider the following task: for $d \times d$ unitary matrices $U$ and $V$, decide whether $UV - VU = 0$ or $\frac{1}{\sqrt{d}} \fnorm{UV - VU} = (\frac{1}{d}\sum_{i,j = 1}^d \abs{\bra{i}(UV - VU)\ket{j}}^2)^{1/2} > \eps$ with success probability $\geq 2/3$.
    Then if this task can be solved with $n$ queries to $\controlled U$ and $\controlled V$, it can also be solved with $n$ queries to $U$ and $V$.
\end{corollary}
\begin{proof}
Let $\rho(U, V)$ be the output of the circuit which solves the distinguishing task using $n$ queries to $\controlled U$ and $\controlled V$.
Because the circuit is allowed measurements and partial traces, we can assume that $\rho(U, V)$ is the output guess: $\rho(U, V)$ is a random classical bit which, with probability $\geq 2/3$, is $0$ if $\frac{1}{\sqrt{d}}\fnorm{UV - VU} = 0$ and $1$ if $\frac{1}{\sqrt{d}}\fnorm{UV - VU} > \eps$.
Then by applying \cref{thm:main} twice, we get a circuit which uses $n$ queries to $U$ and $V$ and outputs $\rho(\varphi U, \phi V)$ for random $\varphi, \phi \in \cyc_\infty$.
Because the non-commutativity measure being tested is invariant under global phase, meaning that $\frac{1}{\sqrt{d}}\fnorm{(\varphi U)(\phi V) - (\phi V)(\varphi U)} = \frac{1}{\sqrt{d}}\fnorm{UV - VU}$, the output of this circuit $\rho(\varphi U, \phi U)$ will be the correct output bit with probability $\geq 2/3$.
\end{proof}


By applying \cref{thm:main} as shown above, we can conclude that many classes of problems do not need controlled access.
The general principle is as follows: sometimes, unitary problems ask for information about the \emph{unitary matrix} $U$; others are more naturally phrased as a problem about the \emph{unitary channel} $\sigma \mapsto U \sigma U^\dagger$.
For problems on unitary matrices, there is a canonical choice of global phase which therefore specifies a particular $\controlled U$, so one can benefit from having controlled access.
For problems on unitary channels, there is no canonical choice of global phase, and so access to any $\controlled(\varphi U)$ should be equally helpful.
But if we are indifferent to the choice of $\controlled(\varphi U)$, a random one should suffice, and \cref{thm:main} shows that this can be simulated with uncontrolled queries.
Problems about \emph{unitary channels} can be decontrolled, while problems about \emph{unitary matrices} may not be.

In other words, \emph{$\controlled U$ does not help for any ``physical'' problem involving $U$}.
These sorts of problems only care about $U$ as a unitary channel, and so do not depend on the global phase of $U$.
As discussed in \cref{cor:state-prep}, this includes problems where $U$ is a state preparation unitary.
We can also make the dual observation, that quantum control doesn't help for unitaries which prepare observables.
In the literature around shadow tomography and gentle measurement, we are interested in understanding POVMs $\braces{A_1, I - A_1}, \dots, \braces{A_m, I - A_m}$, which are typically specified to take the form of $A_i = U_i^\dagger \Pi U_i$, where $\Pi = \proj{0} \otimes I$ and $U_i$ is a unitary we have black-box access to.
(Projector POVMs can be expressed in this form without ancilla, and general POVMs can be expressed in this form with some ancilla.)
Because quantum mechanics is invariant under global phase, so are these observables: $A_i = (\varphi U_i)^\dagger \Pi (\varphi U_i)$ for any $\varphi \in \cyc_\infty$.
So, access to quantum control does not help for these problems, in terms of number of queries to the $U_i$'s.
This includes the quantum OR problem~\cite{hlm17,bwb22} and shadow tomography~\cite{aaronson20,bo24}; algorithms for these problems, some of which use $\controlled U_i$'s, can be decontrolled.

Below, we describe in more detail the generality to which \cref{thm:main} holds, and therefore, the types of problems for which $\controlled U$ does not help.

\begin{enumerate}
    \item The theorem can be repeated, as done in \cref{cor:commutativity}.
        Given a circuit which uses $\controlled U$ and $\controlled V$, we can produce, using only $U$ and $V$, the output of the circuit applied with $\controlled(\phi U)$ and $\controlled(\varphi V)$ for $\phi, \varphi$ independent random phases.
    \item As mentioned in the theorem statement, the result works even for removing controlled access to $U^\dagger$, $U^*$, and $U^\trans$, so controlled unitaries do not help even when these resources are used.
        Note that this is different from applying the theorem to $\controlled U$, $\controlled U^\dagger$, etc.\ individually, since the output simulates the circuit with a random phase which is consistent across the queries: $\controlled (\varphi U)$, $\controlled (\varphi U)^\dagger$, and so on.
    \item The output need not be a single bit.
        For example, the goal of unitary channel tomography is to output an estimate of the unitary channel $\rho \mapsto U \rho U^\dagger$~\cite{hkot23}.
        Since this unitary channel is invariant under the global phase of $U$, $\controlled U$ does not help for this problem.
    \item The output need not be classical.
        For example, for the task of fast-forwarding time evolution, that of preparing $U^{k}\proj{0} U^{-k}$ given $U$, $\controlled U$ does not help.
    \item The input need not be worst case.
        Our result holds for average-case problems too: in \cref{cor:commutativity}, if we only need to be correct over a distribution over $U$ and $V$, then quantum control still would not help.
        In general, quantum control does not help to distinguish unitary ensembles which are phase invariant, meaning that the likelihood of sampling $U$ and $\varphi U$ are equal.
        We will see this being used in our application to PRUs (\cref{cor:pru}).
    \item Even if a circuit needs global phase for a subroutine, if the final outcome is phase-invariant, then controls can be removed.
        For example, the original algorithm for amplitude estimation~\cite{BHMT02} uses phase estimation as a subroutine.
        Phase estimation depends on global phase, so the $\controlled U$'s cannot be removed for that task.
        However, the goal of amplitude estimation is to estimate $\norm{(\proj{0} \otimes I) U \ket{0}}$, which is invariant under the phase of $U$.
        So, the $\controlled U$'s can be removed by running \cref{thm:main} to the full algorithm, not just the phase estimation subroutine.
        We note that it was already known that amplitude estimation does not need $\controlled U$'s~\cite{ar20}, but our work gives a new route for showing this.
\end{enumerate}

\paragraph{Decontrolling to increase efficiency.}
\cref{thm:main} may have utility in settings where we would like to trade gate count for qubit count.
For example, consider a circuit which uses $\controlled U$ $n$ times, where $U$ is a circuit which consists of $T$ two-qubit gates (assuming all-to-all connectivity).
When $U$ is a long quantum circuit, to implement $\controlled U$, every gate of $U$ needs to be replaced with its controlled version.
This introduces overhead: applying the controlled unitaries will use at least $2nT$ two-qubit gates, since each controlled two-qubit gate needs to be compiled down into at least two two-qubit gates.
In fact, the overhead may be more: for example, a controlled CNOT gate (a.k.a.\ a Toffoli gate) requires $5$ two-qubit gates to implement it~\cite{sw95,ydy13}.
On the other hand, the decontrolled circuit will replace these controlled unitaries with the unitary themselves, giving a final gate count of $nT + \bigO{n(\log(n) + \log(d))}$, counting the additional simulation overhead.
When $T$ is large, then the decontrolled circuit is more gate-efficient, at the cost of $\log_2(d)$ qubits of space overhead (using \cref{prop:less-overhead} to reduce the overhead).
The details of this example are not crucial: in general, there may be resources (e.g.\ T-gate count) for which $\controlled U$ costs a constant factor more than $U$; so, whenever the dominant cost of the circuit comes from $U$, it would be worthwhile to perform a decontrolled circuit instead, which incurs a lower-order overhead.
Our result may also be useful for implementing a quantum control qubit in experimental settings (though note that experimental setups can be often modified to embed qubits in a higher-dimensional space, leading naturally to a $\controlled U$ implementation~\cite{Lan+08,zrkzpl11,fddb14}).

\subsection{Which problems do need controlled unitaries?}

We conclude from our theorem that controlled access only helps when the global phase of the unitary matters in some way.
For some problem domains, the global phase is important, and consequently, algorithms for these domains rely on controlled access.
We discuss the most notable instances now.\footnote{
    We should mention a couple omissions---notable algorithms which use controls when they do not need to.
    We have already discussed the algorithms for amplitude estimation~\cite{BHMT02} and the quantum OR problem~\cite{hlm17}: these use controlled $U$, but this dependence can be removed.
    The Hadamard test~\cite{ajl06} estimates the real part of the inner product $\braket{\psi}{\phi}$, given controlled access to state preparation unitaries for $\ket{\psi}$ and $\ket{\phi}$.
    This task depends on the global phase, so the controlled access is necessary, but it is generally used to estimate the phase-invariant $\abs{\braket{\psi}{\phi}}$, which can be done using the control-less swap test~\cite{bcww01}.
    In all of these cases, we could use \cref{thm:main} to remove the controlled queries in these circuits, but other algorithms have also been discovered for these problems which do not use $\controlled U$.
}

First, we consider the most notable algorithm which uses controls: phase estimation~\cite{kitaev95}.
At a high level, this procedure maps an eigenstate of $U$, $\ket{v}$, to $\ket{v}\ket{\lambda}$, where $\lambda$ is an estimate of the corresponding eigenvalue.
This task does depend on the global phase of $U$, since the spectrum of $U$ depends on its global phase.
If we were to apply \cref{thm:main} to this circuit, we would get a procedure which maps $\ket{v} \mapsto \ket{v}\ket{\lambda + \varphi}$ where $\varphi$ is a random phase that is consistent across all of the eigenvalues.
This can be enough for some purposes, as existing work notes~\cite{ccgmss24}.

Second, we consider phase oracles of Boolean functions: for a function $f: [d] \to \braces{0,1}$, its phase oracle is the unitary matrix $Q$ satisfying $Q\ket{x} = (-1)^{f(x)} \ket{x}$.
Since we care about the distinction between $f(x)$ and $1 - f(x)$, which have phase oracles $Q$ and $-Q$, the global phase matters.
(If we do not care about this distinction, then we can simulate $\controlled (\pm Q)$ for a random sign using \cref{thm:main} and \cref{prop:less-overhead}(2).)
This explains why $\controlled Q$ appears in these algorithms.
However, note that we can implement $\controlled Q$ from the (uncontrolled) bit-flip oracle, $O_f \ket{i}\ket{b} \mapsto \ket{i}\ket{b + f(i)}$.
In fact, access to $\controlled Q$ and $O_f$ are equivalent, since they are equal up to conjugation by a Hadamard.
So, if we take the bit-flip oracle as the standard oracle, then controlled access to the phase oracle comes by default.

Third, the linear combinations of unitaries technique, first introduced in the context of Hamiltonian simulation~\cite{cw12}, takes controlled unitaries for $U$ and $V$ and simulates a non-unitary quantum operation of the form $\alpha U + \beta V$~\cite{gslw18}.
This depends on global phase (or at least the relative phase between $U$ and $V$), and so the controlled access cannot be removed.
A related task is to, given a controlled reflection $\controlled R$, measure a state with the POVM $\braces{\Pi, I - \Pi}$, where $\Pi = (I + R)/2$.
This task depends on global phase, in the sense that the POVM labels we attach to $\Pi$ and $I - \Pi$ depend on the eigenvalues of $R$.
Without the labels, i.e.\ if we just wish to apply the map $\sigma \mapsto \Pi \sigma \Pi + (I - \Pi) \sigma (I - \Pi)$, this task can be done with $R$ alone.

Finally, $\controlled U$ queries can help when we are given continuous-time access to a unitary.
In this setting, our access takes the form of $U(t) = e^{-\ii H t}$: we can query $U(t)$ for any $t \in \mathbb{R}$, incurring an evolution time cost of $\abs{t}$.
We can still decontrol circuits which use continuous queries.
For a circuit which calls $\controlled(U(t))$ with a total evolution time of $n$, we can discretize it into one which only calls $\controlled U(\delta)$ and $\controlled U(-\delta) = \controlled U(\delta)^\dagger$ for a sufficiently small $\delta$, and then call \cref{thm:main} to replace these applications with uncontrolled queries.
This works if the original circuit can tolerate Hamiltonians of the form $H + \theta I$ for a potentially large identity term, $\abs{\theta} = \poly(n)$.
Algorithms which use continuous-time controlled queries, like ground state preparation via phase estimation, are typically not designed for such Hamiltonians, though.
Nevertheless, with continuous queries, particularly when the Hamiltonian $H$ is nice, we know that controlled access cannot help much: there are procedures to simulate $\controlled U(t)$ with $U(t)$~\cite{dnsm19,cbw24}, and we can always simulate $\controlled U(t)$ by learning the Hamiltonian and implementing it directly~\cite{blmt24c}.

\subsection{Upgrading the security of pseudorandom unitaries}

Especially following the recent work of Zhandry~\cite{zhandry25},
there has been strengthened interested in showing security of cryptographic primitives against adversaries who have access not just to a cryptographic unitary~$U$, but also to $U^\dagger$, $U^*$, $U^\trans$, and their controlled variants.
One application of our work is that it provides a simple, generic technique for upgrading a cryptographic scheme which is secure against uncontrolled access to~$U$ into a stronger scheme  which is secure also against controlled access.
To illustrate this, we will consider the cryptographic primitive of pseudorandom unitaries, first introduced by Ji, Liu, and Song~\cite{jls18}.

\begin{definition}[Pseudorandom unitaries]
    For each integer $n \geq 1$,
    let $K_n$ be a set of keys
    and $\{U_{n, k}\}_{k \in K_n}$ be a set of $n$-qubit unitaries.
    Then the set of unitaries $\mathcal{U} = \{U_{n, k}\}_{n, k \in K_n}$ is a \emph{pseudorandom unitary (PRU)} if it satisfies the following two properties.
    \begin{itemize}
        \item (Efficient generation): There exists a quantum algorithm which, on input $n \geq 1$ and $k \in K_n$, computes $U_{n, k}$ in time $\poly(n)$.
        \item (Indistinguishability from Haar): For any efficient, $\poly(n)$-time oracle circuit $A(\cdot)$,
        \begin{equation*}
            \Big| \Pr_{k \sim K_n}[\text{$A(U_{n, k})$ accepts}] - \Pr_{U \sim \mathrm{Haar}_n}[\text{$A(U)$ accepts}]\Big| \leq \mathrm{negl}(n),
        \end{equation*}
        where $\mathrm{negl}(n)$ is any negligible function in $n$.
        Here, $\mathrm{Haar}_n$ refers to the Haar measure on $n$-qubit unitaries.
    \end{itemize}
    Typically, the security is proven with respect to adversaries $A()$ which only have oracle access to $U$, in which case $\mathcal{U}$ is a \emph{standard PRU}.
    If it is also secure against adversaries which have oracle access to $U$ and $U^{\dagger}$, then it is a \emph{strong} PRU.
\end{definition}

Ma and Huang~\cite{mh25} showed the existence of strong PRUs assuming standard cryptographic assumptions.
We show that they can be upgraded to even \emph{stronger} PRUs by the addition of a single global phase.

\begin{corollary}[Stronger PRUs] \label{cor:pru}
    Let $\mathcal{U} = \{U_{n, k}\}_{n, k \in K_n}$ be a PRU secure against some subset of $\{U, U^\dagger, U^*, U^\trans\}$ queries.
    Let $q_n = 2^n$, set $K_n^+ = \{(k, \varphi) \mid k \in K_n, \varphi \in \cyc_{q_n}\}$,
    and define $U_{n, (k, \varphi)} = \varphi \cdot U_{n, k}$ for each $(k, \varphi) \in K_n^+$.
    Then the set of unitaries $\mathcal{U}^+ = \{U_{n, (k, \varphi)}\}_{n, (k, \varphi) \in K_n^+}$ is a PRU secure against the corresponding subset of $\{\controlled U, \controlled U^\dagger, \controlled U^*, \controlled U^\trans\}$ queries.
\end{corollary}
\begin{proof}
    Suppose for example that $\mathcal{U}$ is a strong PRU, meaning it is secure against $U$ and $U^\dagger$ queries.
    Then we claim that $\mathcal{U}^+$ is a PRU secure against $\controlled U$ and $\controlled U^\dagger$ queries.
    First, we note that $\mathcal{U}^+$ satisfies the efficient generation property, because for each unitary $U_{n, (k, \varphi)} = \varphi \cdot U_{n, k}$, both $U_{n, k}$ and the global phase $\varphi$ can be computed in $\poly(n)$ time.

    As for the indistinguishability from Haar property, let $A(\cdot)$ be a $\poly(n)$-time oracle circuit which uses queries to $\controlled U$ and $\controlled U^\dagger$. Then there exists a constant $n_0$ such that $A(U)$ makes less than $2^n$ queries for all $n \geq n_0$ and $n$-qubit unitaries~$U$.
    Then \cref{thm:main} gives a $\poly(n)$-time oracle circuit $B(U)$ which uses queries to $U$ and $U^\dagger$ and, on input $U$, accepts with the same probability as $A(\varphi U)$  for a random $\varphi \in \cyc_{q_n}$.
    As a result, we have that for all $n \geq n_0$,
    \begin{align*}
        &\Big| \Pr_{(k, \varphi) \sim K_n^+}[\text{$A(U_{n, (k, \varphi)})$ accepts}] - \Pr_{U \sim \mathrm{Haar}_n}[\text{$A(U)$ accepts}]\Big|\\
        ={}& \Big| \Pr_{k \sim K_n, \varphi \sim \cyc_{q_n}}[\text{$A(\varphi U_{n, k})$ accepts}] - \Pr_{U \sim \mathrm{Haar}_n, \varphi \sim \cyc_{q_n}}[\text{$A(\varphi U)$ accepts}]\Big|\\
        ={}& \Big| \Pr_{k \sim K_n}[\text{$B(U_{n, k})$ accepts}] - \Pr_{U \sim \mathrm{Haar}_n}[\text{$B(U)$ accepts}]\Big| \leq \mathrm{negl}(n).
    \end{align*}
    In the second line, we used that if $U$ is Haar distributed, then so is $\varphi U$ for any phase $\varphi$.
    In the third line, we used the indistinguishability from Haar property of $\mathcal{U}$ and the fact that $B$ is a $\poly(n)$-time oracle circuit.
    Hence, $\mathcal{U}^+$ itself satisfies the indistinguishability from Haar property, which completes the proof.
\end{proof}

We note that this proof crucially relies on the fact that \cref{thm:main} allows us to simulate queries to $\varphi U$ for a uniformly random $\varphi$. This means that the decontrolling strategy of Sheridan, Maslov, and Mosca~\cite{smm09} does not seem to help here, because it is not clear what distribution of phases one should ``upgrade'' the PRU with.
We believe the ability to easily upgrade security proofs is one of the most useful applications of our \cref{thm:main}.
Indeed, our original motivation for studying this problem came from wanting to upgrade the amplitude amplification and estimation lower bounds from our concurrent work~\cite{tw25a} to hold against controlled unitaries; we show how \cref{thm:main} can be used to do so in Remark 3.5 of that work.

\subsection{Relationship to prior work} \label{subsec:prior}

We now survey the existing literature on simulating controlled unitaries using uncontrolled unitaries.

\paragraph{Early results.}
In his paper introducing phase estimation, Kitaev notes a simple method to simulate $\controlled U$~\cite{kitaev95}: if we have a $+1$-eigenstate $\ket{\psi}$ such that $U\ket{\psi} = \ket{\psi}$, then we can implement the controlled unitary (the algorithm: controlled-SWAP with the eigenstate, apply $U$ to the ancilla register, controlled-SWAP).
By the same argument, if we have a general eigenstate $U \ket{\psi} = \lambda \ket{\psi}$, then we can implement $\controlled (\lambda^{-1} U)$ (up to a global phase of $\lambda$).
Since the eigenstate remains unchanged after applying the controlled unitary, we can perform this controlled unitary as many times as we like.

Kitaev's protocol requires knowing an eigenstate, which is hard to find in general.
This issue is surmountable, though.
Sheridan, Maslov, and Mosca observe that by initializing the ancilla to the maximally mixed state, we can interpret this as a random eigenstate, and so we can implement $\controlled (\lambda^{-1} U)$ for a random eigenvalue $\lambda$ of $U$~\cite{smm09}.
This gives a result very similar to ours, with the main difference being that the distribution over $\controlled (\varphi U)$ is not uniform, but instead a distribution which ensures that $\varphi U$ has a $+1$-eigenvalue.
Further, their version has a lower space and gate overhead: see \cref{prop:less-overhead}, where we present this version.
Our result has the nice property of giving a distribution over global phases which does not depend on the unitary being controlled; this allows us to prove an average-case statement like \cref{cor:pru}.
Further, this distribution being a uniformly random phase is the most natural, because of the principle of indifference.

\paragraph{Literature on quantum combs.}
There is a large body of work on higher-order quantum transformations (also called superchannels, supermaps, and combs~\cite{cdp08}), which asks: can we design a circuit which applies $U$ as a black box some number of times to perform some task?
When the task is \emph{controlization}, i.e.\ implementing some $\controlled(\phi U)$, this literature is broadly pessimistic.
Work of Ara\'{u}jo, Feix, Costa, and Brukner~\cite{afcb14} shows that one query to $U$ does not suffice to implement $\controlled (\phi U)$, for any fixed $\phi$ (which is allowed to depend on $U$).
This lower bound informs the rest of this literature, with algorithms for controlization~\cite{fddb14,bdp16,dnsm19,cbw24} only being demonstrated in very limited settings.
There is even a lower bound showing that, for a certain success criterion, implementing some $\controlled(\phi U)$ is impossible, even with arbitrarily many queries to $U$~\cite{gst24}.
Our work does not contradict these lower bounds, since our simulations are fundamentally non-unitary, producing a random $\controlled(\phi U)$ instead of any particular controlled query.
This is a natural guarantee for simulation: the quantum channel corresponding to $U$ does not specify a global phase, so we are happy selecting a random one to controlize.
The key observation we make is that this guarantee suffices for many kinds of problems which use controlled queries.

\paragraph{Simulating Choi states of $\controlled U$.}
Besides the work of Sheridan, Maslov, and Mosca, the most similar prior result we are aware of is contained a lemma of Kretschmer~\cite[Lemma 25]{kretschmer21}.
Kretschmer observes that Choi states of a random $\controlled (\varphi U)$ can be prepared given Choi states of $U$.
The distribution over $\varphi$ is the same as that given in our theorem, so this shows our result in the special case that the controlled queries are only used to prepare Choi states.
This special case suffices to prove \cref{cor:commutativity}, since that particular algorithm takes this form.
However, since Choi states of $V$ are in general much weaker than queries to $V$, it does not suffice for our other implications.


\section{Proofs}

To build intuition for our construction,
we begin by showing how to decontrol the following one-query circuit.
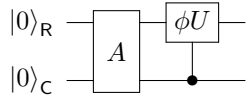
\begin{figure}[h]
\[ 
\Qcircuit @R=1em @C=1em {
\lstick{\ket{0}_{\reg{R}}}
    & \multigate{1}{A}
    & \gate{\phi U}
    & \qw  \\
\lstick{\ket{0}_{\reg{C}}}
    & \ghost{A}
    & \ctrl{-1}
    & \qw
}
\]
\caption{
    A simple one-query circuit.
}
\label{fig:simple}
\end{figure}

\noindent
Letting $\ket{\circuit(\controlled (\phi U))}$ denote the output of this circuit, our goal is to simulate the mixed state
\begin{equation*}
    \E_{\phi \sim \cyc_q} \bracks[\Big]{\proj{\circuit(\controlled (\phi U))}}
\end{equation*}
for some reasonably large number $q$.
In this case, we can write the output state as
\begin{align*}
\ket{\circuit(\controlled (\phi U))}
    = \controlled(\phi U) A \ket{0}
    = (I \otimes \ketbra{0}{0}_{\reg{C}}) \cdot A \ket{0} + (\phi U \otimes \ketbra{1}{1}_{\reg{C}}) \cdot A \ket{0}
    \eqqcolon \ket{\feyn_{0}} + \phi \ket{\feyn_{1}},
\end{align*}
where $\ket{\feyn_{b}}$ is the Feynman path corresponding to the part of the state that queries with the control bit set to~$b$.
Then we can compute the overall mixed state as
\begin{align*}
    \E_{\phi \sim \cyc_q} \bracks[\Big]{\proj{\circuit(\controlled (\phi U))}}
    &= \E_{\phi \sim \cyc_q} \Big[(\ket{\feyn_{0}} + \phi \ket{\feyn_{1}})(\bra{\feyn_{0}} + \phi^{-1} \bra{\feyn_{1}})\Big] \\
    &= \E_{\phi \sim \cyc_q} \Big[\proj{\feyn_0} + \phi^{-1} \ketbra{\feyn_0}{\feyn_1} + \phi \ketbra{\feyn_1}{\feyn_0} + |\phi|^2 \proj{\feyn_1}\Big]\\
    &= \proj{\feyn_0} + \proj{\feyn_1},
\end{align*}
where we used the fact that $\E_{\phi \sim \cyc_q} \phi = \E_{\phi \sim \cyc_q} \phi^{-1} = 0$ once $q \geq 2$.
Thus, the average output is just a mixture over these two Feynman paths, and the reason is that these two paths receive different powers of $\phi$ ($\phi^0$ and $\phi^1$) as their phases.
One way we could simulate this output is to simply measure the controlled bit and apply $U$ to the first wire only if it is 1,
as in the following circuit.
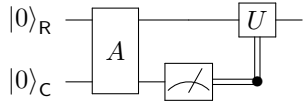
\begin{figure}[h]
\[ 
\Qcircuit @R=1em @C=1em {
\lstick{\ket{0}_{\reg{R}}}
    & \multigate{1}{A}
    & \qw
    & \gate{U}
    & \qw \\
\lstick{\ket{0}_{\reg{C}}}
    & \ghost{A}
    & \meter
    & \cctrl{-1}
}
\]
\caption{
    Simulating the one-query circuit with a mid-circuit measurement.
}
\label{fig:simple-sim}
\end{figure}

\noindent
There are two issues with this construction, however. 
First, one might hope for a circuit which computes a unitary and does not use mid-circuit measurements.
Second, and perhaps more importantly, this circuit does not generalize to the setting of two or more applications of $\controlled(\phi U)$.
To rectify these two issues, we will instead consider the following circuit.
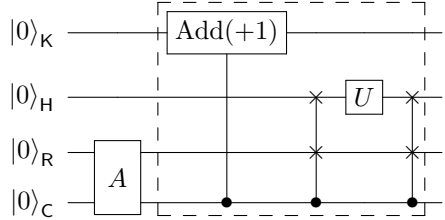
\begin{figure}[h]
\begin{equation*}
    \Qcircuit @R=1em @C=1em {
    \lstick{\ket{0}_{\reg{K}}}
    & \qw
    & \gate{\Add(+1)}
    & \qw
    & \qw
    & \qw
    & \qw \\
    \lstick{\ket{0}_{\reg{H}}}
    & \qw
    & \qw
    & \qswap
    & \gate{U}
    & \qswap
    & \qw \\
    \lstick{\ket{0}_{\reg{R}}}
    & \multigate{1}{A}
    & \qw
    & \qswap \qwx
    & \qw
    & \qswap \qwx
    & \qw \\
    \lstick{\ket{0}_{\reg{C}}}
    & \ghost{A}
    & \ctrl{-3}
    & \ctrl{-1}
    & \qw
    & \ctrl{-1}
    & \qw \gategroup{1}{3}{4}{6}{.7em}{--}
}
\end{equation*}
\caption{Simulating the one-query circuit without mid-circuit measurements.
The area in the dashed box is referred to as the \emph{$\controlled(U)$ gadget}.}
\end{figure}

\noindent
This circuit adds two new registers.
The first register $\reg{K}$ contains a counter which can be incremented.
This register is meant to count the power of $\phi$ along a given Feynman path,
and so it should be incremented when $\controlled(\phi U)$ is queried with the control bit set to 1.
The second register $\reg{H}$ is $d$-dimensional and is a copy of the $\reg{R}$ register. Intuitively, the purpose of this register is the following: when the control bit is set to 1, we apply $U$ to the $\reg{R}$ register. However, when the control bit is set to 0, we still need to apply $U$ to \emph{some} register, and so we apply it to the $\reg{H}$ register instead.
To check that this simulates our one-query circuit, let us first note that the state on the $\reg{R}$ and $\reg{C}$ registers prior to applying the $\controlled(U)$ gadget can be written as
\begin{equation*}
    A \ket{0} = (I_{\reg{R}} \otimes \proj{0}_{\reg{C}}) \cdot A \ket{0}  + (I_{\reg{R}} \otimes \proj{1}_{\reg{C}}) \cdot A \ket{0}.
\end{equation*}
The first term corresponds to the branch of the state with the control bit set to 0 prior to the $\controlled{U}$ query, and the second term corresponds to the branch of the state with the control bit set to 1. Then after the $\controlled(U)$ gadget is applied, the state becomes
\begin{align*}
    &{}(I_{\reg{R}} \otimes \proj{0}_{\reg{C}}) \cdot A \ket{0} \otimes (\ket{0}_{\reg{K}} \otimes U \ket{0}_{\reg{H}})
    + (U_{\reg{R}} \otimes \proj{1}_{\reg{C}}) \cdot A \ket{0} \otimes (\ket{1}_{\reg{K}} \otimes \ket{0}_{\reg{H}})\\
    = & \ket{\feyn_0} \otimes (\ket{0}_{\reg{K}} \otimes U \ket{0}_{\reg{H}}) + \ket{\feyn_1} \otimes (\ket{1}_{\reg{K}} \otimes \ket{0}_{\reg{H}}).
\end{align*}
Tracing out the $\reg{K}$ and $\reg{H}$ now gives us our desired mixture, because the counter $\reg{K}$ contains two different numbers along each path.

Generalizing this construction to multiple $\controlled(\phi U)$ queries is relatively straightforward: we simply replace each instance of $\controlled(\phi U)$ in the original circuit with a separate $\controlled(U)$ gadget in the decontrolled circuit.
Generalizing this further to circuits which have controlled access to $\phi U$, $(\phi U)^\dagger$, $(\phi U)^*$, and $(\phi U)^\trans$ requires a little more care, however.
First, as the counter register $\phi$ is meant to count the power of $\phi$ along a given Feynman path, we will increment the counter when simulating a query to $\phi U$ and $(\phi U)^\trans = \phi U^\trans$, and we will \emph{decrement} the counter when simulating a query to $(\phi U)^{\dagger} = \phi^{-1} U^\dagger$ and $(\phi U)^* = \phi^{-1} U^*$.
This ensures that Feynman paths which have different powers of $\phi$ in the original circuit decohere relative to each other in the decontrolled circuit.
In addition, we will replace the $d$-dimensional $\reg{H}$ register which was initialized to $\ket{0}_{\reg{H}}$ with \emph{two} $d$-dimensional registers $\reg{H}$ and $\reg{H}^\trans$ initialized to the $d$-dimensional EPR state $\ket{\Phi}_{\reg{H},\reg{H^{\trans}}}$.
To simulate a controlled $\phi U$ query when the controlled bit is set to 0, we still apply $U$ to the $\reg{H}$ register, as before.
But if, say, we want to simulate a controlled $(\phi U)^\trans$ query when the controlled bit is set to 0, we will instead apply $U^\trans$ to the $\reg{H}^\trans$ register, giving us
\begin{equation*}
    (I \otimes U^\trans) \cdot \ket{\Phi}_{\reg{H},\reg{H^{\trans}}}
    = (U \otimes I) \cdot \ket{\Phi}_{\reg{H},\reg{H^{\trans}}},
\end{equation*}
the same state as in the $\phi U$ case.
This ensures that the contents of the $\reg{K},\reg{H},\reg{H}^\trans$ registers along a given Feynman path will only depend on the power of $\phi$ that path picks up in the original circuit and not on any other property of that path,
and so two Feynman paths which have the same power of $\phi$ in the original circuit will remain in superposition with each other in the decontrolled circuit.
Together, these modifications suffice to simulate general controlled queries.

\subsection{Proof of \texorpdfstring{\cref{thm:main}}{Theorem 1.1}}

The main theorem is stated to allow for measurements within the circuit, but by the deferred measurement principle, without loss of generality we can consider circuits without interior measurements.
There will be no overhead in assuming this, as we can then recompile the resulting circuit back down to one which uses interior measurements.
We first prove that $\wh{\rho} = \E_{\phi \sim \cyc_q}[\rho(\phi U)]$ for finite $q$; the infinite $q$ case follows similarly.

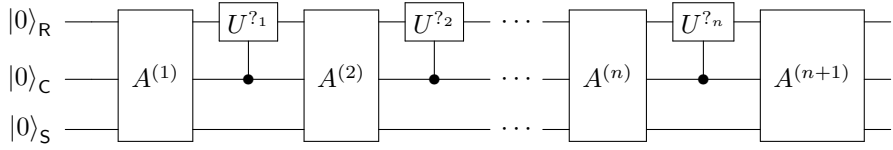
\begin{figure}[ht]
\[ 
\Qcircuit @R=1em @C=1em {
\lstick{\ket{0}_{\reg{R}}}
    & \qw
    & \multigate{2}{A^{(1)}}
    & \gate{U^{\tmpop_1}}
    & \multigate{2}{A^{(2)}}
    & \gate{U^{\tmpop_2}}
    & \qw
    & {\cdots}
    &
    & \multigate{2}{A^{(n)}}
    & \gate{U^{\tmpop_n}}
    & \multigate{2}{A^{(n+1)}}
    & \qw  \\
\lstick{\ket{0}_{\reg{C}}}
    & \qw
    & \ghost{A^{(1)}}
    & \ctrl{-1}
    & \ghost{A^{(2)}}
    & \ctrl{-1}
    & \qw
    & {\cdots}
    &
    & \ghost{A^{(n)}}
    & \ctrl{-1}
    & \ghost{A^{(n+1)}}
    & \qw  \\
\lstick{\ket{0}_{\reg{S}}}
    & \qw
    & \ghost{A^{(1)}}
    & \qw
    & \ghost{A^{(2)}}
    & \qw
    & \qw
    & {\cdots}
    &
    & \ghost{A^{(n)}}
    & \qw
    & \ghost{A^{(n+1)}}
    & \qw  
}
\]
\caption{
    A generic circuit which uses controlled queries to $U$, $U^\dagger$, $U^*$, and $U^\trans$.
    Following \eqref{eq:circuit-def}, the output of this circuit is denoted $\ket{\circuit(\controlled U)}_{\reg{CRS}}$.
}
\label{fig:original}
\end{figure}

An arbitrary quantum circuit which applies $V$, $V^\dagger$, $V^*$, and $V^\trans$ a total of $n$ times can be written in the following form:
\begin{align} \label{eq:circuit-def}
    \ket{\circuit(V)}_\reg{CRS} = A^{(n+1)}_{\reg{CRS}} V^{\tmpop_n}_\reg{CR} A^{(n)}_{\reg{CRS}} \cdots V^{\tmpop_2}_{\reg{CR}} A^{(2)}_{\reg{CRS}} V^{\tmpop_1}_{\reg{CR}} A^{(1)}_{\reg{CRS}} \ket{0}_{\reg{CRS}}.
\end{align}
Here, the $\tmpop$ indicates the type of operation: $V^{\tmpop}$ equals $V$, $V^\dagger$, $V^*$, and $V^\trans$ when $\tmpop = 1$, $\dagger$, $*$, and $\trans$, respectively.
To each type of operation, we associate a sign based on whether the operation inverts phases or not.
\begin{definition}[Sign assigned to matrix operations]
    We define the function $\sigma: \braces{1,\dag,*,\trans} \to \braces{\pm 1}$ such that, for all $\varphi \in \C$ of unit magnitude, $\varphi^{\tmpop} = \varphi^{\sigma(\tmpop)}$. 
    In particular, $\sigma(1) = \sigma(\trans) = +1$, and $\sigma(\dagger) = \sigma(*) = -1$.
\end{definition}
Fix a unitary $U \in \C^{d \times d}$.
Then our goal is to simulate
\begin{align*}
    \E_{\phi \sim \cyc_q} \bracks[\Big]{\proj{\circuit(\controlled (\phi U))}}.
\end{align*}
We begin by writing a particular output of a circuit in terms of Feynman paths.
\begin{definition}[Feynman path through control qubits] \label{def:feyn}
For the quantum circuit $\ket{\circuit(\controlled U)}$ defined in \eqref{eq:circuit-def} and for $0 \leq m \leq n$, we define
\begin{align*}
    \ket{\feyn_{b_1,\dots,b_m}} \coloneqq A^{(m+1)} \parens[\big]{\proj{b_m} \otimes U^{\tmpop_m b_m}} A^{(m)} \cdots \parens[\big]{\proj{b_2} \otimes U^{\tmpop_2 b_2}} A^{(2)} \parens[\big]{\proj{b_1} \otimes U^{\tmpop_1 b_1}} A^{(1)} \ket{0}.
\end{align*}
A circuit diagram for an example of such an expression is given in \cref{fig:feyn}.
As we will see in \cref{claim:feyn}, when $m = n$, this (sub-normalized) state corresponds to the Feynman paths which, on the $i$th controlled query, had the control qubit set to $\ket{b_i}$.
When $m < n$, the state is the same, but for the intermediate state of the circuit after performing $A^{(m+1)}$ and before applying $\controlled U^{\tmpop_{m+1}}$.
Further, we define
\begin{align*}
    \ket{\feyn(k)} \coloneqq \sum_{\substack{b_n,\dots,b_1 \in \braces{0,1} \\ \sigma(\tmpop_n) b_n + \dots + \sigma(\tmpop_1) b_1 = k}} \ket{\feyn_{b_1,\dots,b_n}}
\end{align*}
When the circuit applies only $\controlled U$, i.e.\ all $\tmpop_k = 1$, this is the sum over Feynman paths where $U$ is applied $k$ times.
In general, this is the sum over Feynman paths with a ``weight'' of $k$, where applying $U$ and $U^\trans$ count for weight $+1$ and $U^\dagger$ and $U^*$ count for $-1$.
\end{definition}

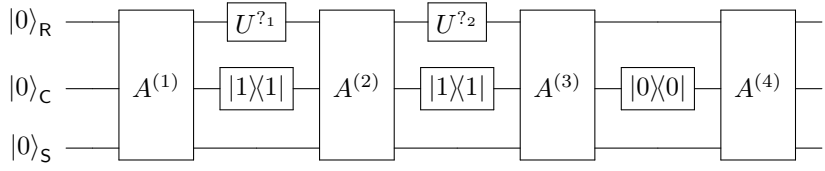
\begin{figure}[ht]
\[ 
\Qcircuit @R=1em @C=1em {
\lstick{\ket{0}_{\reg{R}}}
    & \qw
    & \multigate{2}{A^{(1)}}
    & \gate{U^{\tmpop_1}}
    & \multigate{2}{A^{(2)}}
    & \gate{U^{\tmpop_2}}
    & \multigate{2}{A^{(3)}}
    & \qw
    & \multigate{2}{A^{(4)}}
    & \qw  \\
\lstick{\ket{0}_{\reg{C}}}
    & \qw
    & \ghost{A^{(1)}}
    & \gate{\proj{1}}
    & \ghost{A^{(2)}}
    & \gate{\proj{1}}
    & \ghost{A^{(3)}}
    & \gate{\proj{0}}
    & \ghost{A^{(4)}}
    & \qw  \\
\lstick{\ket{0}_{\reg{S}}}
    & \qw
    & \ghost{A^{(1)}}
    & \qw
    & \ghost{A^{(2)}}
    & \qw
    & \ghost{A^{(3)}}
    & \qw
    & \ghost{A^{(4)}}
    & \qw  
}
\]
\caption{
    The circuit diagram for $\ket{\feyn_{1,1,0}}$, as defined in \cref{def:feyn}.
}
\label{fig:feyn}
\end{figure}

With these definitions, we can decompose the output of the circuit, when the black-box unitary is $\controlled (\varphi U)$.
\begin{claim}[Decomposing a single output into Feynman paths] \label{claim:feyn}
Let $\varphi \in \C$ have unit magnitude.
Then
\begin{align*}
    \ket{\circuit(\controlled (\varphi U))}
    = \sum_{b_n,\dots,b_1 \in \braces{0,1}} \varphi^{\tmpop_n b_n} \cdots \varphi^{\tmpop_1 b_1} \ket{\feyn_{b_1,\dots,b_n}}
    = \sum_{k = -\infty}^\infty \varphi^k \ket{\feyn(k)}
\end{align*}
\end{claim}
\begin{proof}
First, notice that
\begin{align*}
    (\controlled (\varphi U))^{\tmpop}_{\reg{CR}}
    = (\controlled ((\varphi U)^{\tmpop}))_{\reg{CR}}
    = \proj{0}_{\reg{C}} \otimes I_{\reg{R}} + \proj{1}_{\reg{C}} \otimes (\varphi U)^{\tmpop}_{\reg{R}}
    = \sum_{b \in \braces{0,1}} \proj{b} \otimes ((\varphi U)^{\tmpop b}).
\end{align*}
We can use this to expand our expression for the output state.
\begin{align*}
    & \ket{\circuit(\controlled (\varphi U))} \\
    &= A^{(n+1)} (\controlled (\varphi U))^{\tmpop_n} A^{(n)} \cdots (\controlled (\varphi U))^{\tmpop_2} A^{(2)} (\controlled (\varphi U))^{\tmpop_1} A^{(1)} \ket{0} \\
    &= \sum_{b_n,\dots,b_1 \in \braces{0,1}} A^{(n+1)} \parens[\big]{\proj{b_n} \otimes (\varphi U)^{\tmpop_n b_n}} A^{(n)} \cdots \parens[\big]{\proj{b_2} \otimes (\varphi U)^{\tmpop_2 b_2}} A^{(2)} \parens[\big]{\proj{b_1} \otimes (\varphi U)^{\tmpop_1 b_1}} A^{(1)} \ket{0} \\
    &= \sum_{b_n,\dots,b_1 \in \braces{0,1}} \varphi^{\tmpop_n b_n}\cdots \varphi^{\tmpop_1 b_1} A^{(n+1)} \parens[\big]{\proj{b_n} \otimes U^{\tmpop_n b_n}} A^{(n)} \cdots \parens[\big]{\proj{b_2} \otimes U^{\tmpop_2 b_2}} A^{(2)} \parens[\big]{\proj{b_1} \otimes U^{\tmpop_1 b_1} A^{(1)}} \ket{0} \\
    &= \sum_{b_n,\dots,b_1 \in \braces{0,1}} \varphi^{\tmpop_n b_n} \cdots \varphi^{\tmpop_1 b_1} \ket{\feyn_{b_1,\dots,b_n}}
\intertext{
    Finally, we regroup the sum based on the value of $\varphi^{\tmpop_n b_n} \cdots \varphi^{\tmpop_1 b_1} = \varphi^{\sigma(\tmpop_n) b_n + \dots + \sigma(\tmpop_1) b_1}$.
}
    &= \sum_{k=-\infty}^{\infty} \varphi^k \ket{\feyn(k)} \qedhere
\end{align*}
\end{proof}

At this point, we have separated out the dependence on $\varphi$ from the rest of the expression.
We now average over this random $\varphi$ in order to express the output in the following way.

\begin{claim}[Decomposing a mixture of output states into Feynman paths] \label{claim:feyn-mixture}
We can write the output state we wish to simulate in the following way:
\begin{align*}
    \E_{\phi \sim \cyc_q} \bracks[\Big]{\proj{\circuit(\controlled (\phi U))}}
    &= \sum_{k=0}^{q-1} \proj{\feyn(k \mod q)} \\
    & \qquad \text{ where } \ket{\feyn(k \mod q)} \coloneqq \sum_{p = -\infty}^{\infty} \ket{\feyn(k + pq)}.
\end{align*}
\end{claim}
\begin{proof}
Using \cref{claim:feyn}, we have that, for $\varphi \in \cyc_q$
\begin{align*}
    \ket{\circuit(\controlled (\varphi U))}
    = \sum_{k=-\infty}^{\infty} \varphi^k \ket{\feyn(k)}
    = \sum_{k=0}^{q-1} \varphi^k \ket{\feyn(k \mod q)}.
\end{align*}
The second equality holds because $\varphi^q = 1$, so the only possible choices are $\varphi^k$ for $0 \leq k \leq q-1$, and we can combine paths for which $\sigma(\tmpop_n)b_n + \dots + \sigma(\tmpop_1) b_1$ is the same mod $q$.
Hence,
\begin{align*}
    \E_{\phi \sim \cyc_q} \bracks[\Big]{\proj{\circuit(\controlled (\phi U))}}
    &= \E_{\phi \sim \cyc_q} \Bigg[\Bigg(\sum_{k=0}^{q-1} \phi^k \ket{\feyn(k \mod q)}\Bigg)\cdot \Bigg(\sum_{\ell=0}^{q-1} \phi^{-\ell} \bra{\feyn(\ell \mod q)}\Bigg)\Bigg]\\
    &= \sum_{k, \ell=0}^{q-1}\E_{\phi \sim \cyc_q}[\phi^{k-\ell}]\cdot \ketbra{\feyn(k \mod q)}{\feyn(\ell \mod q)} \\
    &= \sum_{k=0}^{q-1} \proj{\feyn(k \mod q)},
\end{align*}
where the last step uses the fact that $\E_{\phi \sim \cyc_q}[\phi^{k-\ell}]$ equals 1 if $k = \ell$ and 0 otherwise. This completes the proof.
\end{proof}

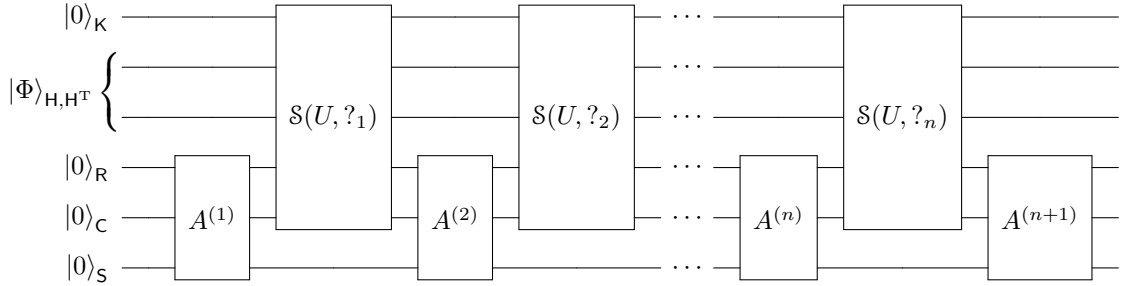
\begin{figure}[ht]
\[ 
\hspace{3em}
\Qcircuit @R=1em @C=1em {
\lstick{\ket{0}_{\reg{K}}}
    & \qw
    & \qw
    & \multigate{4}{\unctrl(U, \tmpop_1)}
    & \qw
    & \multigate{4}{\unctrl(U, \tmpop_2)}
    & \qw
    & {\cdots}
    &
    & \qw
    & \multigate{4}{\unctrl(U, \tmpop_n)}
    & \qw
    & \qw  \\
\lstick{}
    & \qw
    & \qw
    & \ghost{\unctrl(U, \tmpop_1)}
    & \qw
    & \ghost{\unctrl(U, \tmpop_2)}
    & \qw
    & {\cdots}
    &
    & \qw
    & \ghost{\unctrl(U, \tmpop_n)}
    & \qw
    & \qw  \\
\lstick{}
    & \qw
    & \qw
    & \ghost{\unctrl(U, \tmpop_1)}
    & \qw
    & \ghost{\unctrl(U, \tmpop_2)}
    & \qw
    & {\cdots}
    &
    & \qw
    & \ghost{\unctrl(U, \tmpop_n)}
    & \qw
    & \qw
    \inputgroupv{2}{3}{1em}{1em}{\ket{\Phi}_{\reg{H},\reg{H^{\trans}}} \hspace{2em}} \\
\lstick{\ket{0}_{\reg{R}}}
    & \qw
    & \multigate{2}{A^{(1)}}
    & \ghost{\unctrl(U, \tmpop_1)}
    & \multigate{2}{A^{(2)}}
    & \ghost{\unctrl(U, \tmpop_2)}
    & \qw
    & {\cdots}
    &
    & \multigate{2}{A^{(n)}}
    & \ghost{\unctrl(U, \tmpop_n)}
    & \multigate{2}{A^{(n+1)}}
    & \qw  \\
\lstick{\ket{0}_{\reg{C}}}
    & \qw
    & \ghost{A^{(1)}}
    & \ghost{\unctrl(U, \tmpop_1)}
    & \ghost{A^{(2)}}
    & \ghost{\unctrl(U, \tmpop_2)}
    & \qw
    & {\cdots}
    &
    & \ghost{A^{(n)}}
    & \ghost{\unctrl(U, \tmpop_n)}
    & \ghost{A^{(n+1)}}
    & \qw  \\
\lstick{\ket{0}_{\reg{S}}}
    & \qw
    & \ghost{A^{(1)}}
    & \qw
    & \ghost{A^{(2)}}
    & \qw
    & \qw
    & {\cdots}
    &
    & \ghost{A^{(n)}}
    & \qw
    & \ghost{A^{(n+1)}}
    & \qw  \\
}
\]
\caption{
    The circuit diagram for the decontrolled circuit $\ket{\unctrl(\circuit)}$.
    Two ancilla registers are added: register $\reg{K}$ is dimension $n$ and initialized to $\ket{0}$, and registers $\reg{H}$ and $\reg{H^\trans}$ are both dimension $d$, and are together initialized to a maximally entangled state.
    \cref{fig:original} has the diagram of the original circuit $\ket{\circuit}$, and \cref{fig:conversion} contains the diagrams for $\unctrl(U, \tmpop)$.
}
\label{fig:simulation}
\end{figure}

\begin{figure}[ht]
\begin{align*}
\Qcircuit @R=1em @C=1em {
& & & & \lstick{\reg{K}}
    & \gate{\Add(+1)}
    & \qw
    & \qw
    & \qw
    & \qw \\
& & & & \lstick{\reg{H}}
    & \qw
    & \qswap
    & \gate{U}
    & \qswap
    & \qw \\
& \gate{U} & \qw & \push{\rule{.3em}{0em}\mapsto\rule{2em}{0em}} &
    & \qw
    & \qswap \qwx
    & \qw
    & \qswap \qwx
    & \qw \\
& \ctrl{-1} & \qw & &
    & \ctrl{-3}
    & \ctrl{-1}
    & \qw
    & \ctrl{-1}
    & \qw
}
&&
\Qcircuit @R=1em @C=1em {
& & & & \lstick{\reg{K}}
    & \gate{\Add(-1)}
    & \qw
    & \qw
    & \qw
    & \qw \\
& & & & \lstick{\reg{H}}
    & \qw
    & \qswap
    & \gate{U^\dagger}
    & \qswap
    & \qw \\
& \gate{U^\dagger} & \qw & \push{\rule{.3em}{0em}\mapsto\rule{2em}{0em}} &
    & \qw
    & \qswap \qwx
    & \qw
    & \qswap \qwx
    & \qw \\
& \ctrl{-1} & \qw & &
    & \ctrl{-3}
    & \ctrl{-1}
    & \qw
    & \ctrl{-1}
    & \qw
} \\[2em]
\Qcircuit @R=1em @C=1em {
& & & & \lstick{\reg{K}}
    & \gate{\Add(-1)}
    & \qw
    & \qw
    & \qw
    & \qw \\
& & & & \lstick{\reg{H^{\trans}}}
    & \qw
    & \qswap
    & \gate{U^*}
    & \qswap
    & \qw \\
& \gate{U^*} & \qw & \push{\rule{.3em}{0em}\mapsto\rule{2em}{0em}} &
    & \qw
    & \qswap \qwx
    & \qw
    & \qswap \qwx
    & \qw \\
& \ctrl{-1} & \qw & &
    & \ctrl{-3}
    & \ctrl{-1}
    & \qw
    & \ctrl{-1}
    & \qw \\
}
&&
\Qcircuit @R=1em @C=1em {
& & & & \lstick{\reg{K}}
    & \gate{\Add(+1)}
    & \qw
    & \qw
    & \qw
    & \qw \\
& & & & \lstick{\reg{H^{\trans}}}
    & \qw
    & \qswap
    & \gate{U^\trans}
    & \qswap
    & \qw \\
& \gate{U^\trans} & \qw & \push{\rule{.3em}{0em}\mapsto\rule{2em}{0em}} &
    & \qw
    & \qswap \qwx
    & \qw
    & \qswap \qwx
    & \qw \\
& \ctrl{-1} & \qw & &
    & \ctrl{-3}
    & \ctrl{-1}
    & \qw
    & \ctrl{-1}
    & \qw \\
}
\end{align*}
\caption{
    Shown are the outputs of $\unctrl$ applied to controlled queries.
    We ``decontrol'' a general circuit by replacing the controlled queries by their corresponding gadgets (\cref{fig:simulation}).
    All of the gadgets consist of the uncontrolled version of the query, along with two controlled-SWAPs and one controlled add---either an increment or decrement.
    Note that the gadgets for $\controlled U$ and $\controlled U^\dagger$ use register $\reg{H}$, while the gadgets for $\controlled U^*$ and $\controlled U^\trans$ use register $\reg{H}^\trans$.
    When all of these registers are represented with qubits, these additional gates can be implemented using $\bigO{\log(d) + \log(\abs{\reg{K}})}$ two-qubit gates, where $\abs{\reg{K}}$ is the dimension of register $\reg{K}$.
}
\label{fig:conversion}
\end{figure}
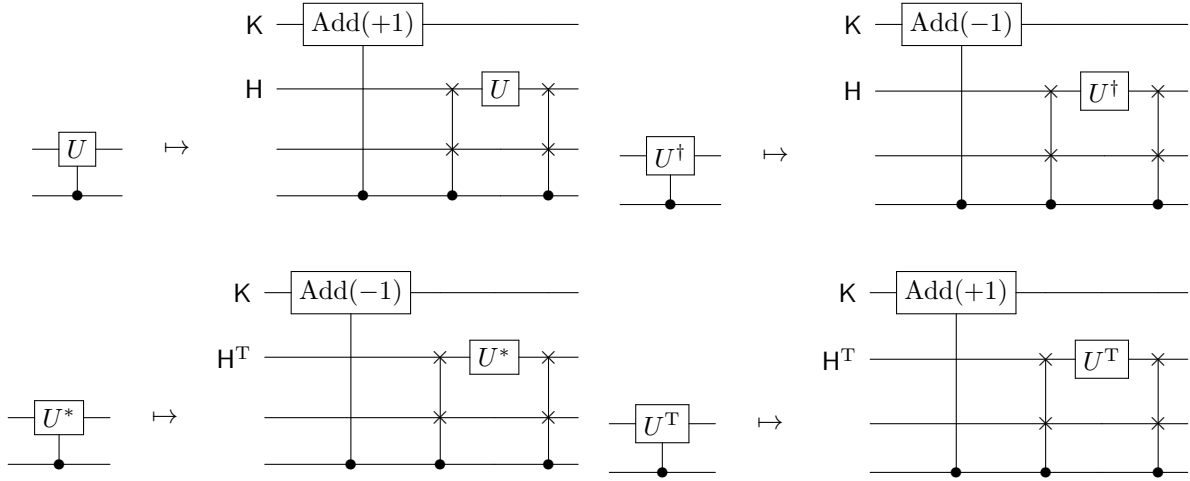

Now, we design a simulation algorithm which produces $\sum_{k=0}^{q-1} \proj{\feyn(k \mod q)}$ as output.
Our strategy is to introduce three registers: a counter register $\reg{K}$ and two hold registers, $\reg{H}$ and $\reg{H^{\trans}}$.
At the start of the circuit, $\reg{K}$ is initialized to $\ket{0}$ and the hold registers are initialized to the maximally entangled state $\ket{\Phi}_{\reg{H},\reg{H^{\trans}}} \coloneqq \frac{1}{\sqrt{d}}\sum_{i=1}^d \ket{i}_{\reg{H}}\ket{i}_{\reg{H^{\trans}}}$.
Then, we replace every $\controlled U^{\tmpop}$ with $\unctrl(U, \tmpop)$, where
\begin{align*}
    \unctrl(U, 1) &= \proj{0}_{\reg{C}} \otimes I_{\reg{R}} \otimes U_{\reg{H}} + \proj{1}_{\reg{C}} \otimes U_{\reg{R}} \otimes \Add(+1)_{\reg{K}} \\
    \unctrl(U, \dagger) &= \proj{0}_{\reg{C}} \otimes I_{\reg{R}} \otimes U^\dagger_{\reg{H}} + \proj{1}_{\reg{C}} \otimes U^\dagger_{\reg{R}} \otimes \Add(-1)_{\reg{K}} \\
    \unctrl(U, *) &= \proj{0}_{\reg{C}} \otimes I_{\reg{R}} \otimes U^*_{\reg{H^{\trans}}} + \proj{1}_{\reg{C}} \otimes U^*_{\reg{R}} \otimes \Add(-1)_{\reg{K}} \\
    \unctrl(U, \trans) &= \proj{0}_{\reg{C}} \otimes I_{\reg{R}} \otimes U^\trans_{\reg{H^{\trans}}} + \proj{1}_{\reg{C}} \otimes U^\trans_{\reg{R}} \otimes \Add(+1)_{\reg{K}}
\end{align*}
Here, $\Add(j)\ket{k} = \ket{k+j}$ is an adder.
Alternatively, we can write the expression in the following way.
\begin{align} \label{eq:unctrl-cases}
    \unctrl(U, \tmpop) = \begin{cases}
        \proj{0}_{\reg{C}} \otimes I_{\reg{R}} \otimes U^{\sigma(\tmpop)}_{\reg{H}} + \proj{1}_{\reg{C}} \otimes U_{\reg{R}}^? \otimes \Add(\sigma(\tmpop))_{\reg{K}}
        & \tmpop \in \braces{1, \dagger} \\
        \proj{0}_{\reg{C}} \otimes I_{\reg{R}} \otimes (U^{\sigma(\tmpop)})^{\trans}_{\reg{H^{\trans}}} + \proj{1}_{\reg{C}} \otimes U_{\reg{R}}^?\otimes \Add(\sigma(\tmpop))_{\reg{K}}
        & \tmpop \in \braces{*, \trans}
    \end{cases}
\end{align}

\cref{fig:conversion} demonstrates how these can be implemented with one query to an uncontrolled unitary.
\cref{fig:simulation} demonstrates the shape of the final simulation circuit, which we denote $\unctrl(\circuit)$:
\begin{align*}
    \ket{\unctrl(\circuit)}
    \coloneqq A^{(n+1)}_{\reg{CRS}} \unctrl(U, \tmpop_n)_{\reg{KHH^{\trans}CR}} A^{(n)}_{\reg{CRS}} \cdots \unctrl(U, \tmpop_2)_{\reg{KHH^{\trans}CR}} A^{(2)}_{\reg{CRS}} \unctrl(U, \tmpop_1)_{\reg{KHH^{\trans}CR}} A^{(1)}_{\reg{CRS}} \ket{0}_{\reg{K}}\ket{\Phi}_{\reg{HH^{\trans}}}\ket{0}_{\reg{CRS}}
\end{align*}
Now, we verify that the output of the simulation circuit equals our desired state.
We will need a basic fact about Choi states.
\begin{fact}[Ricochet property]
    For a matrix $X \in \C^{d \times d}$, let $\ket{\Phi(X)} = \frac{1}{\sqrt{d}}\sum_{i = 1}^d \sum_{j=1}^d X_{i,j}\ket{i}\ket{j}$ be its Choi state.
    Then
    \begin{align*}
        (X \otimes I)\ket{\Phi} = (I \otimes X^\trans)\ket{\Phi} = \ket{\Phi(X)}.
    \end{align*}
\end{fact}

\begin{claim}[Decomposing the simulated circuit into Feynman paths] \label{claim:sim}
The output of the simulation circuit can be written as follows:
\begin{align*}
    \ket{\unctrl(\circuit)}
    &= \sum_{k=-\infty}^{\infty} \ket{\feyn(k)}_{\reg{CRS}} \ket{k}_{\reg{K}} \ket{\Phi(U^{\sigma(\tmpop_n) + \dots + \sigma(\tmpop_1) - k})}_{\reg{HH^{\trans}}}.
\end{align*}
Here, $\ket{k}$ implicitly represents $\ket{k \mod \,\abs{\reg{K}}}$.
\end{claim}
\begin{proof}
We prove this by induction on $n$.
For this proof, we denote $\sigma_m \coloneqq \sigma(\tmpop_m)$ for brevity.
Consider the intermediate state of the circuit $\ket{\unctrl(\circuit)_m}$, after performing $A^{(m+1)}$ and before applying $\unctrl(U, \tmpop_{m+1})$.
We will prove that
\begin{align*}
    \ket{\unctrl(\circuit)_m}
    &= \sum_{k=-\infty}^{\infty} \sum_{\substack{b_m,\dots,b_1 \in \braces{0,1} \\ \sigma_m b_m + \dots + \sigma_1 b_1 = k}} \ket{\feyn_{b_1,\dots,b_m}}_{\reg{CRS}} \ket{k}_{\reg{K}} \ket{\Phi(U^{\sigma_m + \dots + \sigma_1 - k})}_{\reg{HH^{\trans}}}.
\end{align*}
This implies the claim by taking $m = n$, since the states in the $\reg{KHH^\trans}$ registers do not depend on $b_n,\dots,b_1$.

For the base case, $m = 0$, the output state is
\begin{align*}
    \ket{\unctrl(\circuit)_0} = A^{(1)}_{\reg{CRS}} \ket{0}_{\reg{CRS}} \ket{0}_{\reg{K}} \ket{\Phi}_{\reg{HH^{\trans}}} = \ket{\feyn_{\varnothing}}_{\reg{CRS}} \ket{0}_{\reg{K}} \ket{\Phi}_{\reg{HH^{\trans}}}
\end{align*}
which matches as desired.

 For the inductive step, we can write
\begin{align*}
    \ket{\unctrl(\circuit)_{m+1}} &= A^{(m+2)} \unctrl(U, \tmpop_{m+1}) \ket{\unctrl(\circuit)_{m}} \\
    &= \sum_{k=-\infty}^{\infty} \sum_{\substack{b_m,\dots,b_1 \in \braces{0,1} \\ \sigma_m b_m + \dots + \sigma_1 b_1 = k}} A^{(m+2)}_{\reg{CRS}} \unctrl(U, \tmpop_{m+1})_{\reg{KHH^{\trans}CR}} \ket{\feyn_{b_1,\dots,b_m}}_{\reg{CRS}} \ket{k}_{\reg{K}} \ket{\Phi(U^{\sigma_m + \dots + \sigma_1 - k})}_{\reg{HH^{\trans}}}.
\end{align*}
We now consider the summand expression, first looking at the case where $\tmpop_{m+1}$ is either $1$ or $\dagger$.
Then, we can use \eqref{eq:unctrl-cases} to conclude that
\begin{align*}
    & A^{(m+2)}_{\reg{CRS}} \unctrl(U, \tmpop_{m+1})_{\reg{KHH^{\trans}CR}} \ket{\feyn_{b_1,\dots,b_m}}_{\reg{CRS}} \ket{k}_{\reg{K}} \ket{\Phi(U^\ell)}_{\reg{HH^{\trans}}} \\
    &= A^{(m+2)}_{\reg{CRS}}\parens[\Big]{\proj{0}_{\reg{C}} \otimes I_{\reg{R}} \otimes U^{\sigma_{m+1}}_{\reg{H}} + \proj{1}_{\reg{C}} \otimes U_{\reg{R}}^? \otimes \Add(\sigma_{m+1})_{\reg{K}}} \ket{\feyn_{b_1,\dots,b_m}}_{\reg{CRS}} \ket{k}_{\reg{K}} \ket{\Phi(U^\ell)}_{\reg{HH^{\trans}}} \\
    &= \ket{\feyn_{b_1,\dots,b_m,0}}_{\reg{CRS}} \ket{k}_{\reg{K}} U_{\reg{H}}^{\sigma_{m+1}}\ket{\Phi(U^\ell)}_{\reg{HH^{\trans}}} + \ket{\feyn_{b_1,\dots,b_m,1}}_{\reg{CRS}} \parens[\big]{\Add(\sigma_{m+1})\ket{k}}_{\reg{K}} \ket{\Phi(U^\ell)}_{\reg{HH^{\trans}}} \\
    &= \sum_{b_{m+1} \in \braces{0,1}}\ket{\feyn_{b_1,\dots,b_m,b_{m+1}}}_{\reg{CRS}} \ket{k + \sigma_{m+1}b_{m+1}}_{\reg{R}} \ket{\Phi(U^{\ell + \sigma_{m+1}(1-b_{m+1})})}_{\reg{HH^{\trans}}}
\end{align*}
When $\tmpop_{m+1}$ is $*$ or $\trans$, the only difference is that $U^{\sigma_{m+1}}_{\reg{H}}$ is replaced with $(U^{\sigma_{m+1}})^{\trans}_{\reg{H^{\trans}}}$; but these two operations have the same effect on Choi states, so the equality derived above still holds.
We return to the full sum, and substitute the expression we just derived.
\begin{align*}
    &\ket{\unctrl(\circuit)_{m+1}} \\
    &= A^{(m+2)} \unctrl(U, \tmpop_{m+1}) \ket{\unctrl(\circuit)_{m}} \\
    &= \sum_{k=-\infty}^{\infty} \sum_{\substack{b_{m+1}, b_m,\dots,b_1 \in \braces{0,1} \\ \sigma_m b_m + \dots + \sigma_1 b_1 = k}} \ket{\feyn_{b_1,\dots,b_m,b_{m+1}}}_{\reg{CRS}} \ket{k + \sigma_{m+1}b_{m+1}}_{\reg{R}} \ket{\Phi(U^{\sigma_m + \dots + \sigma_1 - k + \sigma_{m+1}(1-b_{m+1})})}_{\reg{HH^{\trans}}} \\
    &= \sum_{k=-\infty}^{\infty} \sum_{\substack{b_{m+1}, b_m,\dots,b_1 \in \braces{0,1} \\ \sigma_{m+1} b_{m+1} + \sigma_m b_m + \dots + \sigma_1 b_1 = k}} \ket{\feyn_{b_1,\dots,b_m,b_{m+1}}}_{\reg{CRS}} \ket{k}_{\reg{R}} \ket{\Phi(U^{\sigma_{m+1} + \sigma_m + \dots + \sigma_1 - k})}_{\reg{HH^{\trans}}}
\end{align*}
This proves the inductive hypothesis.
\end{proof}

Using this claim, we can prove the theorem.
First note that, for $b_n,\dots,b_1 \in \braces{0,1}$, $\sigma(\tmpop_n)b_n + \dots + \sigma(\tmpop_1) b_1$ can only take $n+1$ possible values in an interval.
So, if we take $q \geq n+1$, then $\sigma(\tmpop_n)b_n + \dots + \sigma(\tmpop_1) b_1$ equals $\sigma(\tmpop_n)b_n' + \dots + \sigma(\tmpop_1) b_1' $ if and only if they are equal mod $q$.
In other words,
\begin{align*}
    \E_{\phi \sim \cyc_q} \bracks[\Big]{\proj{\circuit(\controlled (\phi U))}}
    = \sum_{k=0}^{q-1} \proj{\feyn(k \mod q)}
    &= \sum_{k = -\infty}^\infty \proj{\feyn(k)}.
\end{align*}
For the same reason, if $d_{\reg{K}} \geq n+1$, then the vectors $\ket{k}_{\reg{K}}\ket{\Phi(U^{\sigma(\tmpop_n) + \dots + \sigma(\tmpop_1) - k})}_{\reg{HH^{\trans}}}$ are orthogonal for all $k$ which appear in the sum of \cref{claim:sim}.
So,
\begin{align*}
    \tr_{\reg{KHH^{\trans}}}(\proj{\unctrl(\circuit)}_{\reg{KHH^{\trans}CRS}})
    &= \sum_{k = -\infty}^\infty \proj{\feyn(k)}_{\reg{CRS}}.
\end{align*}
The simulation succeeds as desired.

The simulated circuit uses three additional registers, of size $n+1$, $d$, and $d$, respectively.
Each controlled query is replaced with its uncontrolled version, where the gadget uses two controlled SWAPs and one call to an adder modulo $n+1$; these cost $\bigO{\log(n) + \log(d)}$ two-qubit gates per query.

\subsection{Extensions}

\begin{remark}[Handling circuits which query both controlled and uncontrolled versions]
    In our argument, the $A^{(m)}$'s are allowed to contain uncontrolled queries, i.e.\ $U$, $U^\dagger$, $U^*$, and $U^\trans$.
    These queries remain untouched, and the decontrolled circuit still outputs $\E_{\phi \sim \cyc_q}[\rho(\phi U)]$, since for the uncontrolled queries the $\phi$ correspond to a global phase, and so can be added without changing the output.
\end{remark}

\begin{proposition}[Versions of the simulation algorithm with lower space overhead] \label{prop:less-overhead}
    If one wishes to lower the gate and space overhead in \cref{thm:main}, the following modifications can be performed to the simulation algorithm:
    \begin{enumerate}
        \item If the $\reg{K}$ register is removed (along with the gates that interact with it), then the simulation algorithm outputs the mixed state $\frac{1}{d} \sum_{i=1}^d \proj{\circuit(\controlled (\lambda_i^{-1}U))}$, where $\braces{\lambda_i}$ are the eigenvalues of $U$.
        The overhead is then reduced to $\bigO{n \log(d)}$ additional gates and $2\log_2(d)$ additional qubits;
        \item If $U^p = I$, then if we run the simulation algorithm with a counter with $\abs{\reg{K}} = p$, then it outputs $\E_{\phi \sim \cyc_p}[\rho(\phi U)]$, even when $p \leq n$.
        The overhead is then $\bigO{n(\log(p) + \log(d))}$ gates and $\ceil{\log_2(p)} + 2 \log_2(d)$ qubits.
        \item If $\circuit$ only uses controlled $U$ and $U^\dagger$, then the $\reg{H^{\trans}}$ register can be removed from the simulated circuit;
        \item If $\circuit$ only uses controlled $U^*$ and $U^\trans$, then the $\reg{H}$ register can be removed from the simulated circuit.
    \end{enumerate}
\end{proposition}

The first two items in the above proposition remove the $\log(n)$ gate and space overhead in the protocol; the last two reduce the space overhead by $\log_2(d)$ qubits.
Items 1 or 2 can be combined with items 3 or 4: the algorithm with $\reg{K}$ and $\reg{H^{\trans}}$ registers removed is the version stated by Sheridan, Maslov, and Mosca~\cite[Appendix A]{smm09}.
They state the algorithm for $\controlled U$ queries only, but the extension to $\controlled U^\dagger$ is straightforward.
That version has an overhead of $\bigO{n \log(d)}$ gates and $\log_2(d)$ qubits.

Item 2 is particularly relevant when $U$ is taken from a discrete ensemble.
For example, if we know that $U$ is the phase oracle for a Boolean function, so that $U^2 = I$, it may make more sense for an algorithm to have access to a random $\controlled(\pm U)$ than a $\controlled (\phi U)$ for a phase which is random over the entire complex unit circle.

\begin{proof}
To prove item 1, let $U = \sum_{i=1}^d \lambda_i\proj{v_i}$ be the eigendecomposition of $U$.
By an identical argument to \cref{claim:sim}, when $\reg{K}$ is not used, the outcome of the simulation circuit is
\begin{align*}
    \ket{\unctrl_{\reg{-K}}(\circuit)}_{\reg{HH^{\trans}CRS}}
    &= \sum_{k=-\infty}^{\infty} \ket{\feyn(k)}_{\reg{CRS}} \ket{\Phi(U^{\sigma(\tmpop_n) + \dots + \sigma(\tmpop_1) - k})}_{\reg{HH^{\trans}}} \\
    &= \sum_{k=-\infty}^{\infty} \ket{\feyn(k)}_{\reg{CRS}} \sum_{i=1}^d \lambda_i^{\sigma(\tmpop_n) + \dots + \sigma(\tmpop_1) - k}\ket{\Phi(v_iv_i^\dagger)}_{\reg{HH^\trans}} \\
    &= \sum_{i=1}^d \lambda_i^{\sigma(\tmpop_n) + \dots + \sigma(\tmpop_1)} \sum_{k=-\infty}^{\infty} \lambda_i^{-k} \ket{\feyn(k)}_{\reg{CRS}} \ket{\Phi(v_iv_i^\dagger)}_{\reg{HH^\trans}} \\
    &= \sum_{i=1}^d \lambda_i^{\sigma(\tmpop_n) + \dots + \sigma(\tmpop_1)} \ket{\circuit(\controlled (\lambda_i^{-1} U))} \ket{\Phi(v_iv_i^\dagger)}_{\reg{HH^\trans}}
\end{align*}
In the final line, we use \cref{claim:feyn}.
Notice that the $\ket{\Phi(v_iv_i^\dagger)}$'s are orthogonal for $i \in [d]$ with $\braket{\Phi(v_iv_i^\dagger)}{\Phi(v_iv_i^\dagger)} = \frac{1}{d}$.
So, we can conclude that
\begin{align*}
    \tr_{\reg{HH^\trans}}(\proj{\unctrl_{\reg{-K}}(\circuit)}_{\reg{HH^{\trans}CRS}})
    &= \frac{1}{d} \sum_{i=1}^d \proj{\circuit(\controlled (\lambda^{-1}U))},
\end{align*}
as desired.

To prove item 2, we observe that \cref{claim:sim} with $\abs{\reg{K}} = p$ implies that
\begin{align*}
    \ket{\unctrl(\circuit)}
    &= \sum_{k=-\infty}^{\infty} \ket{\feyn(k)}_{\reg{CRS}} \ket{k \mod p}_{\reg{K}}\ket{\Phi(U^{\sigma(\tmpop_m) + \dots + \sigma(\tmpop_1) - k})}_{\reg{HH^\trans}} \\
    &= \sum_{k=0}^{p-1} \ket{\feyn(k \mod p)}_{\reg{CRS}} \ket{k}_{\reg{K}}\ket{\Phi(U^{\sigma(\tmpop_m) + \dots + \sigma(\tmpop_1) - k})}_{\reg{HH^\trans}}
\intertext{
    Here, we use that $U^p = I$, so that the $\reg{HH^\trans}$ registers only depends on the value of $k \mod p$.
    So, we get that
}
    \tr_{\reg{KHH^{\trans}}}(\proj{\unctrl(\circuit)}_{\reg{KHH^{\trans}CRS}})
    &= \sum_{k = 0}^{p-1} \proj{\feyn(k \mod p)}_{\reg{CRS}}.
\end{align*}
This successfully simulates the desired output from \cref{claim:feyn-mixture}.

To prove item $3$, we notice that the simulation only has a gate applied to the $\reg{H}^\trans$ register when $\controlled U^*$ or $\controlled U^\trans$ is in the circuit.
So, if no such queries are made, then we can trace out this register and perform the identical circuit, with the only difference being that $\reg{H}$ is initialized to the maximally mixed state.
An identical argument proves item $4$.
\end{proof}

\section*{Acknowledgments}

E.T.\ is supported by the Miller Institute for Basic Research in Science, University of California Berkeley. 
J.W.\ is supported by the NSF CAREER award CCF-2339711.
We thank Andr\'{a}s Gily\'{e}n, Fermi Ma, William Kretschmer, Ryan O'Donnell, Robin Kothari, Jeongwan Haah, and Rolando Somma for conversations which helped us with the question, ``Is this anything?''

\printbibliography

\end{document}